\documentclass[11pt]{article} 
\usepackage{setspace}

\usepackage{natbib}

\def\bH{\textbf{I}}

\def\bM{\textbf{R}}

\def\cA{\mathcal{A}}

\def\cS{\mathcal{S}}

\def\mE{\mathbb{E}}

\def\smskip{\smallskip}

\def\texitem#1{\par\smskip\noindent\hangindent 25pt
	\hbox to 25pt {\hss #1 ~}\ignorespaces}

\newcommand{\BEAS}{\begin{eqnarray*}}
	\newcommand{\EEAS}{\end{eqnarray*}}
\newcommand{\BEA}{\begin{eqnarray}}
\newcommand{\EEA}{\end{eqnarray}}
\newcommand{\BEQ}{\begin{eqnarray}}
\newcommand{\EEQ}{\end{eqnarray}}
\newcommand{\BIT}{\begin{itemize} \setlength{\itemsep}{0.5cm}}
	\newcommand{\EIT}{\end{itemize}}
\newcommand{\BNUM}{\begin{enumerate}}
	\newcommand{\ENUM}{\end{enumerate}}

\newcommand{\BA}{\begin{array}}
	\newcommand{\EA}{\end{array}}

\makeatletter
\@ifundefined{eqslim}{%
}{}
\makeatother

\usepackage[utf8]{inputenc} 
\usepackage[T1]{fontenc}    
\usepackage{hyperref}       
\usepackage{url}            

\usepackage{booktabs}       
\usepackage{amsfonts}       
\usepackage{nicefrac}       
\usepackage{microtype}      

\usepackage{lmodern}

\usepackage{graphicx}

\usepackage{amssymb}
\usepackage{amsmath}
\usepackage{amsthm}
\usepackage{caption}
\usepackage{geometry}
 \usepackage{algorithm, algpseudocode}

\newcommand{\comment}[1]{}

\let\forallinitial\forall%
\renewcommand{\forall}{\hspace{2mm}\forallinitial\hspace{0.7mm}}

\allowdisplaybreaks

\newtheorem{thm}{Theorem}[section]

\newtheorem{lem}{Lemma}[section]
\newtheorem{defn}{Definition}[section]
\newtheorem*{defn*}{Definition}

\newtheorem*{prop*}{Proposition}

\newtheorem{ass}{Assumption}[section]
\newtheorem{alg}{Algorithm}

{\begin{list}{$\bullet$}{%
			\setlength{\topsep}{0in}
			\setlength{\partopsep}{0in}
			\setlength{\itemsep}{0in}
			\setlength{\parsep}{0in}
			\setlength{\leftmargin}{3.5em}
			\setlength{\rightmargin}{0in}
			\setlength{\itemindent}{-.1in}
		}
	}%
	{\end{list}
}

\DeclareMathOperator*{\argmin}{argmin}
\DeclareMathOperator*{\argmax}{argmax}

\newcommand{\Prop}{ \mathbb{P} }
\newcommand{\Ex}{ \mathbb{E} }

\newcommand{\E} {
	\mathbb{E}}

\title{Robo-advising: Learning Investors' Risk Preferences via Portfolio Choices}

\author{Humoud Alsabah, Agostino Capponi, Octavio  Ruiz Lacedelli, and Matt Stern\thanks{Alsabah: Department of Industrial Engineering and Operations Research, Columbia University, New York, NY 10027; \href{mailto:hwa2106@columbia.edu}{hwa2106@columbia.edu}.
		Capponi: Department of Industrial Engineering and Operations Research, Columbia University, New York, NY 10027; \href{mailto:ac3827@columbia.edu}{ac3827@columbia.edu}.
		Lacedelli: Department of Industrial Engineering and Operations Research, Columbia University, New York, NY 10027; \href{mailto:or2200@columbia.edu}{or2200@columbia.edu}.
		Stern: Department of Industrial Engineering and Operations Research, Columbia University, New York, NY 10027; \href{mailto:mns2141@columbia.edu}{mns2141@columbia.edu}. 
		We are grateful to the participants of the PhD Columbia seminar on Topics in FinTech, the SIAM 2018 Annual Meeting, the SIAM 2019 Conference in Financial Mathematics and Engineering, and the Fields Institute Workshop on Frontier Areas in Financial Analytics. 
	}
}

\date{}

\begin{document}

\maketitle

\begin{abstract}
We introduce a reinforcement learning framework for retail robo-advising. The robo-advisor does not know the investor's risk preference, but learns it over time by observing her portfolio choices in different market environments. {We develop an exploration-exploitation algorithm which trades off costly solicitations of portfolio choices by the investor with autonomous trading decisions based on stale estimates of investor's risk aversion.} We show that the algorithm's value function converges to the optimal value function of an omniscient robo-advisor over a number of periods that is polynomial in the state and action space. By correcting for the investor's mistakes, the robo-advisor may outperform a stand-alone investor, regardless of the investor's opportunity cost for making portfolio decisions. 
\end{abstract}


\section{Introduction}

Robo-advisors have emerged prominently as an alternative to traditional human advisors. First introduced as independent start-ups, with Betterment and Wealthfront being prominent examples, robo-advisors have then been adopted by larger investment companies including, among others, Vanguard and BlackRock. According to \cite{Regan}, robo-advisors managed a total asset value of \$300 billions by 2016, and they are projected to reach \$2.2 trillions only in the United States by 2020.

Using built-in algorithmic procedures, robo-advisors monitor and re-balance investors' portfolios in a low cost and efficient manner (\cite{Bjernes}). They save on fixed costs, such as salaries of financial advisors and maintenance of physical offices, and by reducing investment requirements to a minimum they can charge lower fees. {They provide transparent and systematic advise, mitigating the bias in the data gathering and investors' recommendations process that is typical of human advising (\cite{Linnain}).}

The performance of the robo-advisor strongly depends on its ability to accurately assess the investor's risk tolerance. Current practices followed by robo-advising firms to evaluate investors' risk profiles are based on online questionnaires. 
Wealthfront, a pioneer in the robo-advice space, gauges its investors' risk preferences by asking how they would react to significant losses implied by a market decline and whether they are more interested in maximizing gains, minimizing losses or both equally. On the basis of these answers, as well as other objective metrics (e.g. years to retirement, annual after-tax income to expense ratio), they construct an investment risk metric.
Similar types of questionnaires are also employed by Schwab Intelligent Portfolios, another major automated wealth management service.\footnote{We refer to~\cite{lam2016robo} for additional details on the risk-assessment procedure used by robo-advisor firms as well as the form of interaction with investors to elicit preferences.
}

The use of questionnaires as a means to accurately elicit risk preferences is appealing, but also presents shortcomings. Asking an investor how she would react to a significant loss is unlikely  to account for emotional responses, as the latter are only manifested when the investor actually incurs losses. \cite{holt2002risk} show that investors tend to exhibit more risk tolerance in hypothetical situations as opposed to real ones. \cite{barsky1997preference} assert that survey responses are subject to noise, and hence do not represent an accurate measure of the investor's risk profile. In an experiment, \cite{yook2003assessing} assess the risk tolerance of the same investor using six standard questionnaires. Surprisingly, they find low correlations between the risk tolerance assessment of these questionnaires, even though they were applied to the same investor.\footnote{{The Monetary Authority of Singapore (MAS) issued a consultation paper, in which they propose minimum standard requirements for the regulation of robo-advisor services} (\cite{MAS}). {Those requirements include being able to resolve inconsistent responses from the investor by asking additional questions or contacting her to obtain clarifications on her responses.}} These findings highlight that alternative procedures for risk preference elicitation may be desirable.

{We develop a framework in which the robo-advisor learns the investor's risk preference from experience, i.e., by observing her portfolio choices under changing market conditions.
In each period, the robo-advisor must place the investor's capital into one of several pre-constructed portfolios, each having a distribution of returns that depends on the prevailing market condition. Each portfolio decision reflects the robo-advisor's belief on that specific investor's risk preferences. The robo-advisor can update its assessment by asking the investor to make the portfolio selection herself. Soliciting portfolio choices, however, presents an opportunity cost to the investor who needs to devote time to perform market research or seek more expert advice.
This leads to an \emph{exploration/exploitation tradeoff}: The robo-advisor must decide between
making investments decisions based on its current estimate of the investor's risk preference or soliciting a costly action from the investor that can improve its risk preference assessment.

We propose a planning algorithm, and show that it converges to the (intractable) optimal investment policy in a ``small'', i.e., polynomial in the quantities describing the system, number of steps. Our convergence analysis is inspired from the PAC-MDP (Probably Approximately Correct in Markov Decision Processes) approach.\footnote{See \cite{strehl2009reinforcement} for a survey on PAC-MDP algorithms and their corresponding complexity bounds.} By taking advantage of structural properties of our robo-advising framework, we show that tighter bounds, relative to those implied by existing PAC-MDP algorithms, on the number of convergence steps can be achieved. Specifically, we exploit the key property that our investor is ``small'' relative to the market she trades in, and thus investment decisions executed by the robo-advisor on her behalf do not affect future market conditions.

We demonstrate that the convergence rate to the optimal policy depends on the consistency of the investor's portfolio choices. The robo-advisor takes longer to learn the risk preference of an investor who commits many mistakes when selecting her portfolio (e.g. unsure of her own risk preference, lacking composure, and subject to variations in emotions) as opposed to an investor who acts more consistently with her risk preferences (e.g.  more self-confident on her risk-tolerance, disciplined, and with low emotional variations). {We calibrate our model to business cycle data from the National Bureau of Economic Research (NBER)
and Vanguard's economic and market outlook reports.}\footnote{Vanguard currently holds the largest robo-advisor in the world, with \$112 billions assets under management.} {Our calibration shows that the robo-advising system achieves a higher value of the investment criterion than that of an investor-only model,} in which the investor directly chooses the portfolio at a cost. The avoidance of these opportunity costs is one of the major advantages of robo-advising, which allows the investor to delegate time-consuming activities to the algorithm and considerably reduce these costs.} {A second important advantage of the robo-advising system is its ability to correct for the investor's mistakes. Our analysis shows that even if the portfolio choices were costless to the investor, she would still be outperformed by a robo-advisor with limited initial knowledge on the investor's risk preference.}

{The paper proceeds as follows. Section~\ref{sec:litreview} relates our paper to existing literature.} Section~\ref{sec:model} formally describes the model. Section~\ref{sec:near_opt_sol} presents the learning algorithm and provides worst-case theoretical guarantees. Section~\ref{sec:nums} {analyzes the learning speed of the investor's risk preferences and the value of the robo-advising system for a calibrated version of our model.} {Section \ref{sec:conclusion} concludes the paper.} Technical proofs and additional supporting material are delegated to the Appendix.

\section{Literature Review} \label{sec:litreview}

Our work contributes to the so far scarce literature on robo-advising. \cite{d2019promises} empirically investigate the implications of a robo-advising platform on performance and trading behavior of an investor. {In their framework, the robo-advisor selects an optimal portfolio for each investor, and {the investor} has the option to accept or override this decision. They find that the adoption of robo-advising increases portfolio diversification and reduces well-known behavioral biases, such as the disposition effect, trend chasing, and the rank effect. 
\cite{Das1} {develop a framework for goals-based wealth management.\footnote{The goals-based investment strategy is also followed by Betterment, one of the leading robo-advisor firms with \$14 billions assets under management. Betterment accounts for the investor's time horizon and attributes like age, retirement time, annual income, and investment goals. However, it does not account for an investor's subjective risk tolerance in the portfolio selection procedure, and employs mean–variance optimization to construct efficient portfolios.} They constrain the set of admissible portfolios to those that lie on the Markowitz's efficient frontier, and allow the investor to specify her goal in terms of desired probabilities of achieving the goal. Using geometric arguments, they show that the simultaneous achievement of all specified goals pin down the desired portfolio on the efficient frontier.}
Unlike these studies, our work views the robo-advisor as an algorithm that provides systematic advice, and calibrates itself to the risk profile of the investor it serves. Specifically, {our approach elicits information about the investor by offering her a discrete catalogue of portfolios, that may be viewed as lying on the efficient frontier.}

{
	 Our work is closely related to the stream of literature on
	 inverse reinforcement learning (IRL) (\cite{russell1998learning}, \cite{ng2000algorithms}, \cite{abbeel2004apprenticeship}), which aims at learning an agent's reward function by observing her behavior. These studies typically assume that demonstration trajectories are exogenously specified, or that inquiring for an agent's action is costless (see \cite{arora2018survey} for a recent survey). As in the IRL literature, the goal of our robo-advisor is to learn the investor's risk preference and hence her reward function. Unlike studies in IRL, in our framework the robo-advisor endogenously determines the amount of information to elicit from the investor so to strike a balance between {the opportunity cost faced by the investor and the suboptimality of decisions based on stale information}. This is highly relevant in retail robo-advising, where soliciting an action from the investor requires close monitoring of financial markets and hence presents an opportunity cost.  
}

{Our work is also related to} studies {that} utilize reinforcement learning to maximize an investor's risk-adjusted returns. In a multi-arm bandit setting, \cite{sani2012risk} study the problem of finding the arm that provides the best risk-return trade-off by using variance as a measure of risk.
They propose two algorithms to solve this mean–variance bandit problem and provide theoretical bounds on regret. In a similar setting, \cite{vakili2016risk} provide algorithms that guarantee tighter bounds on regret than the ones in \cite{sani2012risk}.
In these studies, the risk-aversion parameter of the investor is assumed to be known while the market model is not. Specifically, the mean and variance of returns for each portfolio is unknown, and the objective is to maximize rewards, given a pre-specified investor's risk preference. By contrast, in our framework the robo-advisor is unsure about the investor's risk preference, but has no uncertainty about the market model.

{We conclude by mentioning that, in addition to helping investors to design and execute investment portfolios in accordance with their risk appetite, robo-advisors also offer other types of wealth management services. One of the most important services provided is automated tax management. The early work of} \cite{Constantinides} {finds that, in the presence of taxes, the optimal stock trading strategy is to realize capital losses immediately and defer capital gains for as long as possible. Such a strategy lies at the heart of loss-harvesting, as described by} \cite{Stein}. The laborious work of tax-loss harvesting requires continuous monitoring, which is ideally suited for robo-advisors that significantly cut down the time and labor costs for these activities. Hence, robo-advisors increase the number of opportunities to successfully harvest a tax-loss, relative to traditional investment managers (\cite{Traff}).


\section{Model}\label{sec:model}
{We consider an investment horizon consisting of $T$ periods. There exists a set of $m$ pre-specified investment portfolios. The investment decisions are delegated to a robo-advisor, which learns the investor's risk preference over time, and at each time selects the portfolio {whose returns distribution} best reflect the learned preferences. Throughout the paper, we use $\bH$ to denote the investor, and $\bM$ to denote the robo-advisor.} 

\subsection{System States}

{The system states model the market environment, assumed to be represented by the distribution of portfolio returns in each state of the market. Formally, $\cS = \{s^{(1)},\dots,s^{(n)}\}$ is the set of economic scenarios. We use~$X_{s,p^{(i)}}$ to denote the random return} of portfolio $i$, denoted by $p^{(i)}$, in state $s \in \cS$. For example, $s=s^{(1)}$ may correspond to a bearish market environment characterized by low return and volatility, while $s=s^{(n)}$ may indicate a bullish market scenario characterized by high returns and high volatility. 
Note that the distribution of returns for each portfolio is time invariant, i.e., it depends only on the prevailing economic state at time~$s$, but not on  time~$t$ directly. The probability of a transition from state $s$ to state $s'$ is assumed to be independent of the investment decision, and is denoted by $\Prop(s' \mid s)$ for all $s,s' \in \cS$. This means that the portfolio choice does not influence the market environment. We denote by~$s_t$ the prevailing economic state at time~$t$.

\subsection{Robo-advisor Action Set}\label{sec:actions}

We denote the set of actions available for $\bM$
by $\cA^{\bM} = \{ask\}\cup\{p^{(1)},\dots,p^{(m)}\}$. An action~$a_t^\bM = ask$ corresponds to asking the investor to make a portfolio decision and $a_t^\bM = p^{(i)} \in \cA^{\bM}$ corresponds to $\bM$ choosing portfolio $i$ at time $t$. If solicited, the investor chooses an action from the set $\cA^{\bH} = \{p^{(1)},\dots,p^{(m)}\}$ where $a_t^\bH = p^{(i)}$ means that the investor selects portfolio~$i$ at $t$. {If not solicited, the investor does nothing, which we denote by~$a_t^\bH=\{null\}$.}
Active intervention by the investor is costly. We assume such a cost to be constant and denote it by $\kappa>0$. {This cost reflects the effort applied by the investor to choose a portfolio, including close monitoring of financial markets and solving her own optimization problem. The attention span required to make investment decisions on short time-scales needs to be taken out from other activities, and thus presents an opportunity cost to the investor.}

\subsection{Investor's Behavior and Optimization Criterion}\label{sec:investor}

Empirical evidence (see, for instance, \cite{Rachleff2014} and \cite{bucciol2018financial}) suggests that investors' risk preferences are positively correlated to market performance.
These claims are further supported by Betterment, as its clients tend to increase their portfolio risk following periods of strong market performance (\cite{Swift2015}). To capture these empirical patterns, we make the investor's risk preference depend on the prevailing market environment. Specifically, in each state~$s$, the investor's risk-aversion is quantified by the parameter~$\theta_s$, which belongs to a finite set~$\Theta$.
We denote the profile of the investor's risk aversion across all states by~$\theta:=\left(\theta_{s^{(1)}},\dots, \theta_{s^{(n)}}\right)$. We assume that the machine does not know the value of $\theta$ at the beginning of the investment period.

{At time~$t$, we denote the investor's utility function by~$u(\theta_{s_t},{s_t},a_t)$, where we recall that ${s_t}$ is the market state at $t$, $\theta_{s_t}$ is investor's risk aversion in that state, and~$a_t$ is the portfolio chosen at $t$. }
At any given point in time, an investor may not {act consistently with her risk preferences}.
We employ the following model for the investor's behavior: \\

\noindent {\bf Investor's Choice Model:} At any time~$t$, the investor acts according to the risk preference~${\tilde{\theta}_t \sim \mathcal{B}_{s_t}}$ where $\mathcal{B}_{s_t}$ is a probability distribution on~$\Theta$ with a mean of $\theta_{s_t}$. \\

\noindent This model assumes that although the investor commits mistakes when making decisions, she is on average correct, i.e., acts according to her true risk preferences on average. {This is consistent with empirical evidence} (\cite{schildberg2018risk}), which {suggests that lack of discipline, stress, and temporary emotions cause variation in risk preferences around the average level.} 
The robo-advisor does not know the distribution~$\mathcal{B}_{s_t}$ used by the investor to draw decisions, but knows the support of the distribution and the variance of mistakes in each state. It also knows that the investor will act according to her true risk preference on average.\footnote{The analysis presented in this paper can be easily generalized to the case where {the investor is not correct on average, i.e., she} consistently overestimates or underestimates her risk preference.} 

{The objective of the robo-advisor is to maximize the investor's utility when it does not know the {investor's} risk preference a priori, but learns it over time by observing her portfolio choices. The robo-advisor faces a tradeoff between soliciting costly actions from the investor that allow for portfolio choices more tailored to the investor's risk aversion, and acting based on stale estimates of her risk aversion}. 
Specifically, if the current economic state is~$s_t$ and the investor is solicited, she selects a portfolio that is optimal according to her {current utility function} {characterized by the} risk aversion parameter~$\tilde{\theta}_{t}$. This  
means that the investor chooses~${a_t^\bH=g(\tilde{\theta}_{t},{s_t})}$ where
\begin{align*}
g(\tilde{\theta}_{t},{s_t}) := \argmax_{a \in \cA^{\bH}}  u(\tilde{\theta}_{t},{s_t},a).
\end{align*}

\begin{ass}\label{ass:invertable}
The optimal portfolio function~$g(\cdot,{s_t})$ is invertible.
\end{ass}
\noindent Assumption~\ref{ass:invertable} implies that, by observing the investor's action~$a_t^\bH$, the robo-advisor can infer the corresponding risk aversion parameter~$\tilde{\theta}_{t}$. Consider, for example, a portfolio selection procedure in which a share~$w$ of wealth is invested in the risky asset and the remaining share~$(1-w)$ in the risk-free asset. For example, Betterment allows its clients to choose any portfolio along the efficient frontier, by asking them to select their desired combination of stocks and treasury securities (\cite{Holeman2018}). {A sufficient condition for Assumption}~\ref{ass:invertable} to be satisfied {is that the investor's share~$w$ of the risky asset be strictly monotone in her risk aversion parameter~$\theta_s$; a relation we expect to hold in practice since the higher the investor's risk aversion, the lower the amount she will invest in the risky asset.}

Observe that even though the robo-advisor can imply a risk aversion parameter from the investor's portfolio selection, there is no guarantee that this reflects the true investor's risk aversion. The investor acts consistently with her risk preferences only on average, and her mistakes prevent the robo-advisor from learning the actual risk aversion parameter in each market state from a single observation of the investor's choice in that state.

Next, we discuss an important family of utility models, known as \emph{mean-risk}. The mean-risk approach, first introduced by \cite{Markowitz}, postulates an objective function of the form
	\begin{equation}
	{u(\theta_{s_t}, s_t,a_t)=\E[X_{s_t,a_t}]- \theta_{s_t}\mathbb{D}[X_{s_t,a_t}]},
	\end{equation}
	where $\E$ denotes the expectation operator, and $\mathbb{D}$ is measure of dispersion capturing the uncertainty of portfolio returns. Hence, {the above criterion trades off expected returns with the amount of undertaken risk, and the risk-aversion coefficient~$\theta_{s_t}$ determines the balance. The larger $\theta_{s_t}$, the more the investor penalizes risky portfolio choices. Risk functionals typically used in applications include:}
	\begin{itemize}
		\item \emph{Mean-variance}, which is obtained by setting
		\begin{equation}\label{eq:mean_variance}	
		\mathbb{D}[X_{s_t,a_t}]:=Var[X_{s_t,a_t}].
		\end{equation}
		{This model is the most popular asset allocation framework employed by robo-advising firms. For example, Wealthfront uses mean-variance optimization in its purest form, and Schwab complements mean-variance analysis with full-scale optimization (}\cite{lam2016robo}).\footnote{Full-scale optimization considers all features of the return distribution, including skewness and kurtosis. In Schwab's full-scale optimization, the disutility from losses is weighted twice as large as the utility of an equal-sized gain. The reader is referred to~\cite{adler2007mean} {for a more detailed comparison of mean-variance and full-scale optimization.}} 

		\item \emph{Central semideviations}, obtained by choosing 
			$$\mathbb{D}[X_{s_t,a_t}]:= \left(\E\left[\left(\E[X_{s_t,a_t}]-X_{s_t,a_t}\right)_+^p\right]\right)^{1/p},$$
		where $p \in [1,\infty)$ is a fixed parameter and $(x)_+:=\max(x,0)$. {Unlike mean-variance, which weights equally upward and downward deviations from the mean, this measure only penalizes downward deviations of~$X_{s_t,a_t}$ from the mean.}
		
		\item \emph{Weighted mean deviations from quantiles}, which is obtained by choosing
			\begin{equation*}\label{eq:avg_mean_dev}
			\mathbb{D}[X_{s_t,a_t}]:= \E\left[\max \{ (1-\alpha)(H^{-1}_{X_{s_t,a_t}}(\alpha)-X_{s_t,a_t}),\alpha(X_{s_t,a_t}-H^{-1}_{X_{s_t,a_t}}(\alpha)) \} \right],
			\end{equation*}
		where $\alpha \in (0,1)$ is a fixed parameter and $H^{-1}_{X_{s_t,a_t}}(\cdot)$ is the left-size $\alpha$-quantile of the $H_{X_{s_t,a_t}}(x)=\Prop(X_{s_t,a_t}\leq x)$, that is, $H^{-1}_{X_{s_t,a_t}}(\alpha):= \inf \{x:H_z(x)\geq \alpha\}$. {Unlike the previous two models, which penalize deviations from the average, this criterion measures the tails on the far left of returns captured by the quantile parameter~$\alpha$.}\footnote{The mean-deviation from quantiles model is closely related to Value-at-Risk (VaR) and Conditional Value-at-Risk (CVaR). Denote by~$q_{\alpha}$, $VaR_{\alpha}$, and $CVaR_{\alpha}$ the mean-deviation, VAR, and CVaR at the $\alpha$ quantile. Then, the following identity holds
		$$CVaR_{\alpha}(X)=\frac{1}{\alpha}\int_{1-\alpha}^{1}VaR_{1-\tau}(X)d\tau= \E[X]+\frac{1}{\alpha}q_{1-\alpha}[Z].$$
		We refer to~\cite{shapiro2009lectures} for the details.	
		} 	
	\end{itemize}

{The above discussed models differ in the distributional content needed for their implementation. The mean-variance is the least demanding, just requiring mean and variance of returns, while the weighted mean deviations from quantiles is the most demanding, requiring the full distribution of returns $X_{s,a}$ for each state~$s$ and action~$a$.}

\subsection{Robo-Advisor Policy}\label{sec:robo-advisor}

Denote the set of public histories by
\begin{align*}
H_{t} := \bigcup\limits_{k=1}^t \left(\mathcal{A}^{\bH}\times\mathcal{A}^{\bM}\right)^{k-1}\times\mathcal{S}^{k},
\end{align*}
where $h_{t} = \left(s_{1},a^{\bH}_{1}, a^{\bM}_{1}, \ldots, a^{\bH}_{t-1}, a^{\bM}_{t-1}, s_{t}\right) \in H_{t}$ for $t> 1$ and $h_1 = s_1$. A public history contains information that is observed by both the investor and the robo-advisor. This includes the realization of the system's states and the actions executed by both agents. At each time~$t$, the robo-advisor chooses an action~$a^{\bM}_{t}= \pi_t(H_{t})$, where the policy~$\pi_t$ is adapted to the history $H_{t}$.
{It is crucial that the robo-advisor's policy depends on the whole history, because prior states and actions convey information to the robo-advisor about the investor's risk preference.}  Our framework is closely related to inverse reinforcement learning, where the agent aims to estimate the reward function from previous expert demonstrations (see Definition 2 in~\cite{arora2018survey}). We collect the sequence of policies adopted by the robo-advisor from period~$1$ through~$T$ in the vector~$\pi=(\pi_1,\dots,\pi_T)$.

At any time~$t$, the investor's reward when the robo-advisor takes action~$a^{\bM}_{t}$ is given by
\begin{align*}
r(s_t, a^{\bM}_t,\theta_{s_t}) = \begin{cases}
{u(\theta_{s_t},{s_t},a_t^\bH)} - \kappa,  & \text{if  }  a^{\bM}_t= ask, \\
{u(\theta_{s_t},{s_t},a_t^\bM)}, & \text{otherwise},
\end{cases}
\end{align*}
{where we recall that $\kappa$ is the intervention cost incurred by the investor.}
If~$\theta=\{\theta_s\}_{s \in S}$ is known {to the robo-advisor}, {we can find the optimal value function~$V^*$ by solving the \emph{Bellman's optimality equation}}
\begin{equation*}
V^{\ast}_{\tau}(s,\theta)=\max_a r(s, a,\theta_s)+\sum_{s^\prime}\Prop(s^\prime|s)V^{\ast}_{\tau-1}(s^\prime,\theta),
\end{equation*}
where~$\tau$ is the number of remaining periods. 
The robo-advisor cannot achieve an optimal reward because the investor's risk preference vector~$\theta$ is not known a priori. The robo-advisor's objective is to find the policy that maximizes the investor's cumulative reward, i.e.,
\begin{align}\label{eq:robo-advisor_obj}
\pi^*=\argmax_\pi \E\left(\sum_{t=1}^T r(s_t, a^{\bM}_t,\theta_{s_t} ) \right),
\end{align}
where~$a^{\bM}_t=\pi(H_t)$ and the expectation is taken with respect to the probability distribution of the state {trajectory}~$(s_1,s_2,\dots, s_T)$.\footnote{We assume that $T$ is finite and the discount factor is~$\gamma=1$ to simplify exposition. All results easily extend to the case where $T=\infty$ and $\gamma<1$ (see Appendix~\ref{app:infinite_horizon} for the details).} We remark here that
maximizing~\eqref{eq:robo-advisor_obj} is a notoriously difficult task, and is tractable only in few highly specialized cases.\footnote{The $k$-armed bandit problem is one notable exception. In this problem, the well-known \emph{Gittins indices} may be used to solve for the optimal policy in a Bayesian setting. We refer to \cite{gittins2011multi} for further details.} Therefore, we consider the relaxed goal of acting near-optimally on all but a polynomial number of steps.
More formally, if~$\mathcal{P}_t$ is the policy followed by the algorithm at time~$t$, let $V^{\mathcal{P}_t}_{\tau}({h_t},\theta):=\Ex\left(\sum_{j=t}^{t+\tau-1} r(s_j, a^{\bM}_j,\theta_{s_j}) | h_t \right)$ where $a^{\bM}_j=\mathcal{P}_j(h_j)$.\footnote{The value function $V^{\ast}_{\tau}(s,\theta)$ achieved by an omniscient robo-advisor, which knows the investor's risk preference~$\theta$ in advance, only depends on the current state~$s_t$. This due to Markov property and the fact that an omniscient robo-advisor does need to keep track of the investor's previous actions to infer her risk preference. However, $s_t$ is replaced with~$h_t$ in the value function~$V^{\mathcal{P}_t}_{\tau}({h_t},\theta)$ since the robo-advisor in our framework typically does not know investor's risk preference~$\theta$ in advance and needs to estimate it from her previous actions (i.e., the policy $\mathcal{P}_t$ is non-stationary). 
}
We require that with a probability greater than~$1-\delta$,
$$ V^{\mathcal{P}_t}_{\tau}({h_t},\theta) \geq V^{\ast}_{\tau}(s_t,\theta) - \epsilon$$
for all but~$m=O(poly(|S|, \tau, 1/\epsilon, 1/\delta))$ time steps. This performance measure guarantees that with a high probability, a near-optimal solution is attained in a ``small'' (i.e., polynomial) number of steps.
Note that in a finite number of steps, no algorithm can identify the optimal policy with probability one. First, there is no guarantee that all market states will be visited in a finite number of steps. Second, even if all states were to be reached, the algorithm is still unable to infer the true investor's risk preference due to investor's mistakes in her portfolio choices. Thus, a failure probability of at most~$\delta$ is allowed. We refer to \cite{kakade2003sample} for a more detailed description of this performance measure and its uses.

\section{Near-Optimal Rewards}\label{sec:near_opt_sol}

In this section, we propose an algorithm that achieves near optimal rewards in a small, i.e., polynomial in various quantities describing the system, number of steps with a high probability. The speed of learning depends on the consistency of investor's decisions, quantified by the variance of the distribution of mistakes, denoted by~$\sigma_s^2$, and range of its support, denoted by~$R_{s}$. It also depends on the required estimation accuracy of the risk aversion parameter. Both~$\sigma_s^2$ and~$R_{s}$ {are} formally defined in Section~\ref{sec:sample_complexity}.


\subsection{Algorithm}\label{sec:alg}

Denote by~$a_{n,s}^{\bH}$ the~$n$-th action executed by the investor in state~$s$. Denote by~$\hat{\theta}_s^{(n)}$ the estimate of the investor's risk-aversion parameter in state~$s$ after she has been observed acting~$n$ times in that state. This is computed as
\begin{align}\label{eq:avg_risk_aversion} \hat{\theta}_s^{(n)}:=\frac{g^{-1}(a_{1,s}^{\bH};s)+g^{-1}(a_{2,s}^{\bH};s)+\dots+g^{-1}(a_{{n},s}^{\bH};s)}{n}.
\end{align}
{This} average estimator may be computed incrementally. Given the current estimate~$\hat{\theta}_s^{(n-1)}$ and after observing the~$n$th investor's action  $a_{n,s}^{\bH}$, the revised estimate of the risk-aversion parameter $\hat{\theta}^{(n)}_s$ is, for $n\geq 2$,
\begin{align*}
\hat{\theta}^{(n)}_s	&=\frac{1}{n}\sum_{i=1}^{n} g^{-1}(a_{i,s}^{\bH};s) \\
&=\hat{\theta}_s^{(n-1)}+\frac{1}{n}\left(g^{-1}(a_{n,s}^{\bH};s) -\hat{\theta}_s^{(n-1)}\right).
\end{align*}
This iterative formula (see Appendix~\ref{app:incremental} for the detailed derivation) holds even when~$n = 1$, yielding $\hat{\theta}^{(1)}_s=g^{-1}(a_1^{\bH};s)$ for any~$\hat{\theta}^{(0)}_s$.\footnote{A naive, yet conceptually easier, implementation of~\eqref{eq:avg_risk_aversion} would keep track of investor's actions, across all states, and then re-compute this average from scratch whenever the risk aversion parameter estimate is needed.  However, {the} computational and memory requirements would grow over time as more investor's actions are solicited, an important concern for robo-advisors which are designed to serve a large number of clients concurrently.
}

\begin{alg}\label{alg:}
	Input $C(s)$ for all~$s\in S$. Set $N(s)=0$ and~$\hat{\theta}_s=0$ for all~$s\in S$.
	{Let~$s_1$ be prevailing economic state at time 1.} For~$t=1,2, \dots, T$:
	\begin{enumerate}
		\item If $N(s_t)<C(s)$, set~$a_t^{\bM}:= ask$ and $N(s_t)=N(s_t)+1$. Observe the investor's action~$a_t^{\bH}$, and update as $\hat{\theta}_{s_t}=\hat{\theta}_{s_t}+\frac{g^{-1}(a_t^{\bH};s_t)-\hat{\theta}_{s_t}}{N(s_t)}$.
		\item Otherwise, set~$a_t^{\bM}:=\argmax_a r(s_t,a,\bar{\theta})$ where $\bar{\theta}=\argmin_{x \in \Theta} |x - \hat{\theta}_{s_t} |$.
	\end{enumerate}
\end{alg}

Algorithm~\ref{alg:} takes as input a \textbf{sample complexity} parameter~$C(s)$, which 
indicates the number of {solicitations} the robo-advisor makes to the investor in state~$s$ before it starts investing autonomously. We say that the investor's risk preference in state~$s$ is {\textbf{accurately estimated}} if~$C(s)$ {solicitations} {have been} made. 
Intuitively, an accurate risk preference estimate~$\hat{\theta}_s$ requires a {large value of}~$C(s)$. {We provide a formal analysis of how to pick~$C(s)$ to guarantee convergence} in Section~\ref{sec:sample_complexity}.

\subsection{Sample Complexity Analysis}\label{sec:sample_complexity}

\begin{defn}
	The \textbf{sample complexity function}, $f(s,\delta)$, is the minimum number of {solicitations} such that if~$d\geq f(s,\delta)$ investor's actions in state~$s$ have been observed, then $\theta_s = \argmin_{x \in \Theta} |x - \hat{\theta}_s|$ with a probability at least~$1-\delta$ regardless of the true investor's risk aversion parameter~$\theta_s$.
\end{defn}

The following lemma provides a bound on the performance of the algorithm when the risk aversion parameter~$\theta_s$ is accurately estimated for a subset~$K$ of the states. It shows that the rewards obtained from the algorithm are {close} to {the} optimal rewards, provided that the probability of escaping the subspace~$K$ is low.

\begin{lem}\label{lemma:escape_bound}
	Let~$\mathcal{P}_t$	denote the policy followed by Algorithm~\ref{alg:} at time~$t$, $s_t$ be the economic state at $t$, and set~$C(s)=f\left(s,\frac{\delta}{|S|}\right)$. 
Let $\theta$ be the true risk aversion parameter of the investor, $K_{t}$ be the set of states with {accurately estimated} risk preferences at time~$t$, and~$A(K_{t},s_t)$ be the event that a state not in~$K_{t}$ is visited in the remaining~$\tau$ steps when starting from state~$s_t \in K_{t}$. Then, with probability at least~$1-\delta$,
$$V^{\mathcal{P}_t}_{\tau}({h_t},\theta) \geq V^{\ast}_{\tau}(s_t,\theta) - 2\tau r_{max}\Prop(A(K_{t},s_t)) $$
where~$|r(\cdot)| \leq r_{max}$. 
\end{lem}

Next, we present our main theoretical result, which
gives a performance guarantee for the algorithm.


\begin{thm}\label{thm:complexity}
Let~$\mathcal{P}_t$	denote the policy followed by Algorithm~\ref{alg:} at time~$t$, $s_t$ be the {economic} state at $t$, and {set}~$C(s)=f\left(s,\frac{\delta}{|S|}\right)$.
Let~$\theta$ be the true investor type. Then with probability at least~$1-2\delta$,
$$V^{\mathcal{P}_t}_{\tau}({h_t},\theta) \geq V^{\ast}_{\tau}(s_t,\theta) - \epsilon $$
--- i.e., the algorithm is $\epsilon$-close to the true optimal policy --- for all but
$$ m= O\left(\frac{\tau^2  r_{max} \sum_s C(s)}{\epsilon}\ln\frac{1}{\delta} \right) $$
time steps.
\end{thm}
Theorem~\ref{thm:complexity} states that, with high probability, our algorithm performs near optimally for all but a “small” number of time steps — where “small” here means polynomial in various quantities describing the MDP.
We observe that the provided bound is tighter than standard PAC-MDP bounds. To the best of our knowledge, therein the constant is
$$ m= \tilde{O}\left(\frac{\tau^6  r_{max}|S||\mathcal{A}|}{\epsilon^2} \right) $$
time steps, where~$|\mathcal{A}|$ is the number of actions available in each state.\footnote{The $\tilde{O}(\cdot)$ notation suppresses logarithmic factors.} This bound is achieved, for instance, by~\cite{kolter2009near}, which consider optimality with respect to the Bayesian policy for a given belief state, rather than the optimal policy for {a} fixed MDP. To achieve this bound, they assume that the prior on the unknown model parameter follows a Dirichlet distribution. In contrast, we adopt a frequentist approach which thus does not require the specification of any prior distribution. We are able to obtain a tighter bound by exploiting a unique feature of our robo-advising framework: investment decisions do not impact future market conditions. This stands in contrast with standard PAC-MDP analysis in other application domains, where actions in a given state typically influence state transitions.


In the remainder of this section, we examine how different distributions of investor mistakes~$\mathcal{B}_s$ impact the sample complexity function~$f(s,\delta)$. Let~$\Theta=\{\theta^i\}_{1\leq i\leq |\Theta|}$ be the set of possible risk aversion parameters.
Without loss of generality, assume~$\theta^{i+1} > \theta^i$ for any~$\theta^i \in \Theta$ and set~$\xi:= \min_{1\leq i \leq | \Theta|-1} \theta^{i+1} - \theta^i$. The following lemma provides an upper bound on the sample complexity function for each state~$s$.

\begin{lem}\label{lem:complexity_func}
	If the distribution of the investor's mistakes~$\mathcal{B}_s$ has variance~$\sigma_s^2$, then the sample complexity function~$f(s,\delta)$ is at most $1+ \frac{4\sigma_s^2}{\delta \xi^2}$. 	
\end{lem}

The upper bound given in Lemma~\ref{lem:complexity_func}
depends on two model parameters: \emph{consistency of investor's decisions}~(measured by the 	
variance of investor's mistakes~$\sigma_s^2$) and the {\emph{resolution of the risk aversion grid}}~($\xi$). If the investor's portfolio decisions are consistent (i.e., low~$\sigma_s^2$), then a low sample complexity parameter is required to guarantee an accurate estimate of the risk aversion parameter.
In fact, as evident from Lemma~\ref{lem:complexity_func}, in the limiting case where the investor does not make mistakes (i.e., $\sigma_s=0$), only one solicitation is needed to infer her risk aversion parameter in state~$s$. The {resolution} parameter~$\xi$ controls the required precision of the risk aversion parameter estimate. As the needed precision grows to infinity (i.e.,~$\xi$ approaches zero), greater exploration and learning is necessary. The following lemma provides an alternative bound on the sample complexity function, which depends on the range of support~$R_s$ of the distribution~$\mathcal{B}_s$, as opposed to the variance~$\sigma_s^2$.\footnote{{The range of support, $R_s$, is defined as the difference between the maximum and minimum values of the support of the distribution $\mathcal{B}_s$. Recall that the support of a distribution is the set of all values attainable by the random variable.}}

\begin{lem}\label{lem:complexity_func_R}
	If the support of the distribution of investor's mistakes~$\mathcal{B}_s$ has a range~$R_s$, then the sample complexity function~$f(s,\delta)$ is at most~$\frac{2R_s}{\xi^2}\ln \frac{2}{\delta}$.
\end{lem}

Lemma~\ref{lem:complexity_func_R} gives qualitative insights similar to Lemma~\ref{lem:complexity_func}, except that the consistency of the investor's decisions is {now measured by the range of mistakes rather than their variance.} Note that this bound increases only logarithmically in~$1/\delta$, and thus {becomes tighter than the bound} provided in Lemma~\ref{lem:complexity_func} as~$\delta \rightarrow 0$.
Because the bounds given in this section
apply to any distribution of mistakes, they are very conservative. As we show in Section~\ref{sec:nums}, in practice, the sample complexity function is much lower than the theoretical bounds.

\section{Calibration}\label{sec:nums}

\subsection{Calibration Procedure}\label{sec:empirical_calibration}

The calibration procedure picks model parameters that are implied by business cycle data released by NBER and Vanguard's economic and market outlook reports. {We choose the mean-variance defined in}~\eqref{eq:mean_variance} as the investor's risk criterion.  
The robo-advisor investment decisions are made on a monthly basis. Accordingly, all calibrated parameters are estimated in monthly units.

\begin{itemize}
	\item \textbf{System states:} There are three states, representing low, medium, and high return volatility markets. The average length of the U.S. business cycle
	during the period $1945-2019$ is 70 months.\footnote{The US business cycle data is obtained from the NBER (\href{https://www.nber.org/cycles.html}{https://www.nber.org/cycles.html})} Because there are three states and each of them is visited twice in a business cycle, each state lasts approximately $12$ months. A period of 12 months in a state corresponds to an escape probability of $0.08$. Accordingly, we set the probability that a state stays the same from period~$t$ to~$t+1$ to $0.92$. If the medium return volatility market is currently prevailing, then with a probability of $0.04$ the state will transition to the high return volatility market and with a probability of $0.04$ the state will transition to the low return volatility market in the next period. If the low (high) return volatility market is currently prevailing, then with a probability of $0.08$ the state will transition to the medium return volatility market in the next period.
	
	\item \textbf{Portfolio selection:} There is a risky and a risk-free asset. We use the monthly treasury rate as a proxy for {the return} of the risk-free asset and set~$r=0.2\%$ for the monthly return.\footnote{The monthly treasury rate is as of 05/16/2019. Source: {U.S. Department of the Treasury} website: \href{https://www.treasury.gov/resource-center/data-chart-center/}{https://www.treasury.gov/resource-center/data-chart-center/}.} The returns and standard deviation of the risky asset depend on the prevailing market conditions. According to Vanguard's economic and market outlook reports, predicted yearly returns for the equity market take values between 6\% and 15\%.\footnote{Vanguard's economic and market outlook reports can be found at their news center\linebreak website: \href{https://pressroom.vanguard.com/}{https://pressroom.vanguard.com/}.} Therefore, in the low, medium, and high return volatility market, we respectively set the monthly returns to 0.5\%, 0.875\%, and 1.25\%. We choose the monthly standard deviation for the medium return volatility market to be 4\%. This is based on the sequence of historical prices of the S\&{P} 500 in the previous 5~years, starting from July 11th 2014 through July 2nd 2019, which yielded monthly returns approximately the same as those in the medium return volatility market. Using the medium-volatility market as a benchmark, we set the monthly standard deviations of the low and high return volatility {market} to 3\% and 5\%, respectively. Portfolio selection corresponds to choosing a share~$w$ {of wealth} to invest in the risky asset and~$(1-w)$ in the risk-free asset. Once the weight~$w$ is specified by the investor, the mean and standard deviation of the corresponding portfolio in a given state~$s$ are
	$\mu(s,w)=w\mu_s+(1-w)r_s$ and $\sigma(w)=w\sigma_s$, respectively,
	where $r$ is the average return of the risk-free asset, $\mu_s$ is the average return of the risky asset in state~$s$, and $\sigma_s$ is the standard deviation of the risky asset in state~$s$. We assume no short selling is allowed, i.e., $0\leq w \leq 1$. This choice is consistent with current practices of robo-advising firms, which only invests cash from investors, that is, leveraging and shorting is not allowed. 
We limit the set of possible weight choices to~$\{0.0001, 0.0002. \dots, 1\}$ to obtain a discrete set of portfolios.
	
	\item \textbf{Investor behavior:} We assume that a typical retail investor allocates at least 20\% of her portfolio to stocks. This range covers income-oriented investors (20\%-40\% stock share) that seek current income with minimal risk, balanced-oriented investors (40\%-70\% stock share) that seek to reduce potential volatility but are willing to tolerate some risk, and growth-oriented investors (70\%-100\% stock share) that aim to maximize returns with little concern on risk taking.\footnote{Vanguard: \href{https://personal.vanguard.com/us/insights/saving-investing/model-portfolio-allocations}{https://personal.vanguard.com/}.} Accordingly, the set of possible risk aversion parameters for the retail investor is \linebreak$\Theta=\{2.2,2.3, \dots, 8.2,8.3\}$, where the lowest value in the set
	corresponds to the risk aversion parameter for which it is optimal to choose a portfolio with 100\% shares in stocks (i.e., $w=1$) 
	and the largest value
	corresponds to the parameter for which it is optimal to choose a portfolio with 20\% shares in stocks (i.e., $w=0.2$). The set~$\Theta$ includes risk aversion parameters of the same magnitude as those estimated empirically in \cite{bucciol2011household}.\footnote{\cite{bucciol2011household} quantify risk using a \emph{risk tolerance} parameter, which is just the reciprocal of our risk-aversion parameter~$\theta$.} Note that the {resolution} parameter in this setup is~$\xi=0.1$. In Appendix~\ref{app:granularity}, we show that using a finer grid to define the investor's risk preference does not {increase much the rewards from investment decisions}. If~$\theta_{s_t}$ is the investor's true risk preference {in state $s_t$}, she will act as
	\begin{equation}\label{eq:mistakes}
	\tilde{\theta}_t=random.sample(\{\theta_{s_t}-r,\dots,\theta_{s_t},\dots,\theta_{s_t}+r\}),
	\end{equation}
	where~$r$ {bounds the size of mistakes} and~$random.sample(\cdot)$ is a function that samples uniformly at random from the input set.\footnote{If some of the elements in the set~$\{\theta_{s_t}-r,\dots,\theta_{s_t},\dots,\theta_{s_t}+r\}$ do not belong to~$\Theta$, we truncate the distribution of mistakes  at the edges to ensure that all elements in the set belong to~$\Theta$ while preserving the ``correct on-average'' requirement ~$\E(\tilde{\theta}_t)=\theta_{s_t}$.} The function $random.sample(\cdot)$ is the maximum entropy probability distribution for a given support (i.e., it assumes the least amount of information on~$\theta$), and hence the numerical results in this section provide a conservative assessment of the learning speed. For the baseline case, we set~$r=3$. We test the sensitivity of our results to different values of~$r$ in Section~\ref{sec:calibration_results}.

	\item \textbf{Opportunity cost:} \cite{Keyes2019} reports that as of January 2019, Betterment manages \$15 billion in assets for its 400,000 customers. Based on this, we estimate the average investment in robo-advising firms to be \$37,500.  Moreover, the U.S. Census Bureau reports that the median household income in the U.S. is \$61,372.\footnote{This is based on the latest income report issued by the U.S. Census Bureau, which was for the year 2017. The full report can be found in their website: \href{https://www.census.gov/library/publications/2018/demo/p60-263.html}{https://www.census.gov/library/publications/2018/demo/p60-263.html}.} This yearly salary corresponds to an hourly rate of \$29.51. If we assume that the average investor spends 1 hour per month in researching and deciding on portfolios, then the opportunity cost is precisely this hourly rate. However, because returns are measured relative to the total invested amount, this opportunity cost must also be measured relative to such an amount. Therefore, dividing the {hourly} rate by the average investment amount, we obtain an {opportunity} cost~$\kappa=0.08\%$. Observe that this value is based on the assumption that each investor spends only 1 hour per month making investment decisions, and hence is rather conservative. In Section~\ref{sec:calibration_results}, we demonstrate that the value added by robo-advising increases as~$\kappa$ increases.


\end{itemize}

\subsection{Calibration Results}\label{sec:calibration_results}

{This section presents the result of our calibrated model. We investigate the rate at which the machine learns the investor's risk preferences, and the value of the robo-advisor.}
We assess the performance of our algorithm using Monte-Carlo simulations. 
In every {run}, the investor's risk-aversion parameter~$\theta_s$ in each state~$s$ is sampled uniformly at random from the set~$\Theta$.
{We measure the value added by the robo-advisor by comparing its investment performance to that of an investor-only model in which the investor makes decisions in isolation.}
In the investor-only model, the investor needs to choose her own portfolio every month and incur the cost~$\kappa$ for doing so. We also compare the performance of the robo-advisor algorithm with that of an omniscient {robo-advisor}, {which} knows the investor's risk preference~$\theta$ in advance. The {omniscient robo-advisor} model provides an upper bound for achievable rewards.

\begin{figure}
	\begin{center}
		\includegraphics[scale=1]{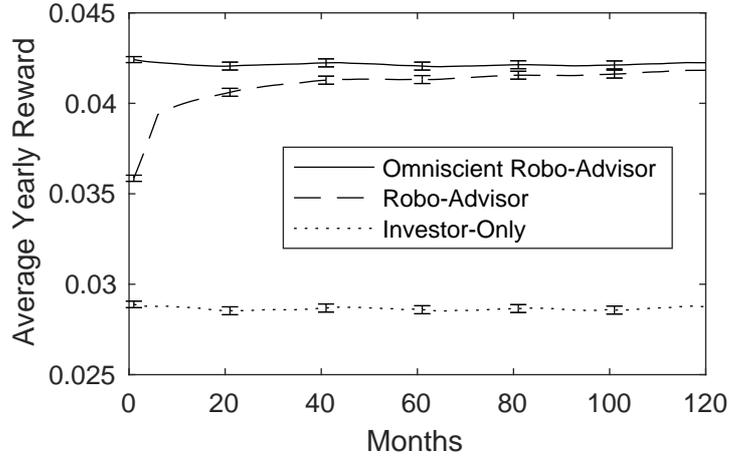}
		\caption{Yearly reward versus time, averaged over $10^4$ trials and shown with 95\% confidence. The {bound on the size} of mistakes is~$r=3$, the sample complexity is~$C=5$, the opportunity cost is~$\kappa=0.08\%$, and the total learning period is $10$ years (i.e., 120 months). The yearly performance of the {robo-advisor} model becomes close to {the} maximum achievable rewards {obtained by an omniscient robo-advisor} within a period of two years. The value added by the robo-advisor is evident from the noticeable large gap in its achieved yearly rewards relative to those of the {investor}-only model.}  \label{fig:learning_speed}
	\end{center}
\end{figure}

Throughout the section, we set the sample complexity function to be equal across states, 
and we denote it simply by~$C$. Figure~\ref{fig:learning_speed} shows the average yearly reward versus time for the omniscient {robo-advisor}, {robo-advisor}, and {investor-only} models. It shows that the yearly rewards achieved by the {robo-advisor} model become very close to those attainable by an omniscient {robo-advisor} within a period of two years. Therefore, as suggested by our theoretical results, the proposed algorithm converges fast to the optimal policy. It is worth noting that the sample complexity parameter is set to $C=5$,
hence much lower than what is required by Theorem~\ref{thm:complexity}. This is explained by the fact that our {theoretical result} considers worst-case bounds, and hence the sample complexity of the exploration bounds, such as those in Lemmas~\ref{lem:complexity_func} and~\ref{lem:complexity_func_R}, may be overly conservative in practice.

\begin{figure}
	\begin{center}
		\includegraphics[scale=1]{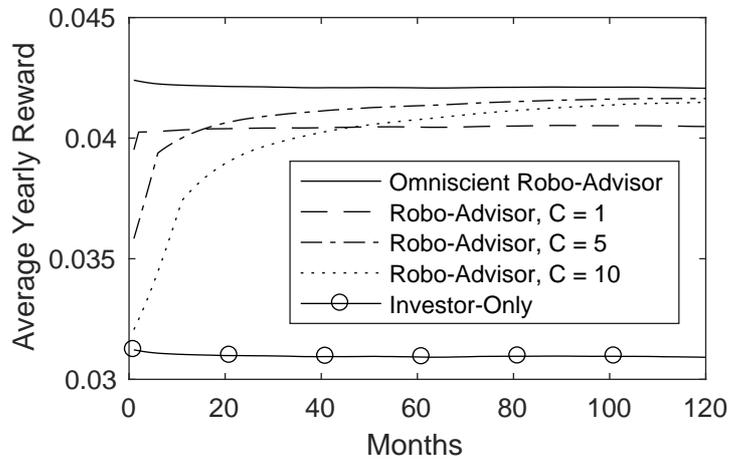}
		\caption{Yearly reward versus time, averaged over $10^4$ trials. The bound on the size of mistakes is~$r=3$, the opportunity cost is~$\kappa=0.08\%$, and the total learning period is $10$ years (i.e., 120 months). The trade-off between exploration and exploitation is captured by the sample complexity parameter~$C$. Increasing the sample complexity~$C$ reduces the short-term rewards but provides longer-term learning benefits.
}  \label{fig:complexity_h}
	\end{center}
\end{figure}

In Figure~\ref{fig:complexity_h}, we assess the performance of the robo-advisor algorithm using different sample complexity parameters. The choice of sample complexity governs the trade-off between shorter-term reduced payoff and longer-term learning benefits. When the sample complexity parameter is low, more exploitation is performed at the expenses of learning, which increases short-term rewards but prevents the robo-advisor from accurately learning the true risk preference of the investor. As the sample complexity parameter increases, short-term rewards are sacrificed to improve learning, providing a better performance in the longer term.

\begin{figure}
	\begin{center}
		\includegraphics[scale=1]{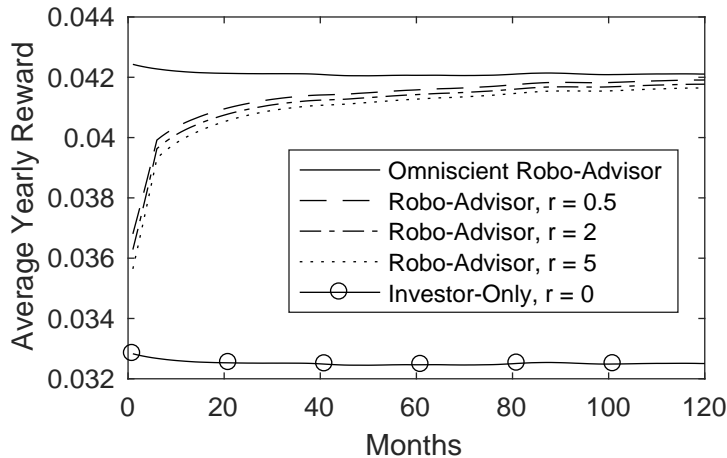}
		\caption{Yearly reward versus time, averaged over $10^4$ trials. The sample complexity is~${C=5}$, the opportunity cost is~$\kappa=0.08\%$, and the total learning period is $10$ years (i.e., 120 months). The performance of the robo-advisor improves as the investor provides more consistent investment decisions (i.e., {the size of mistakes is low}). Even if the investor makes {large} mistakes in the robo-advisor model, the yearly performance exceeds that of the {investor}-only model with no mistakes.
		} \label{fig:mistakes_h}
	\end{center}
\end{figure}

The speed of learning is also impacted by the size of mistakes~$r$, as defined in~\eqref{eq:mistakes}. Consistently with Lemma~\ref{lem:complexity_func_R}, Figure~\ref{fig:mistakes_h} highlights that mistakes slow down the learning process. However, even if the size of the investor's mistakes is large, the robo-advisor model still outperforms the investor-only model \emph{even} {in the hypothetical case where the investor makes no mistakes when she acts independently.} This is because the investor does not incur the opportunity cost~$\kappa$ if the portfolio is chosen by the robo-advisor.
By contrast, this cost is incurred every period in the investor-only system. From an operational perspective, this is one of the primary advantages of robo-advising, in that it allows the investor to delegate research on investment instruments, times for portfolio re-balancing, and other time-consuming activities to the robo-advisor.

\begin{figure}
	\begin{center}
		\includegraphics[scale=1]{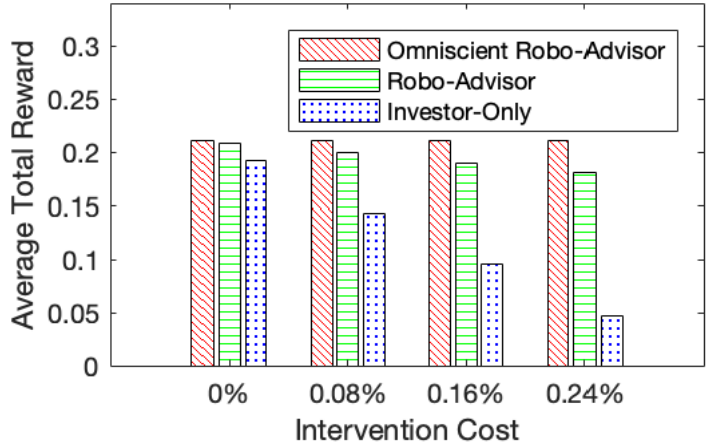}
		\caption{Expected total reward of the omniscient {robo-advisor} model (red), the {robo-advisor} model (green), and the investor-only model (blue), as a function of the opportunity cost~$\kappa$. The total reward is over a period of~$5$ years, averaged over~$10^4$ trials. The bound on the size of mistakes is~$r=3$ and the sample complexity is~$C=5$. The higher is the opportunity cost~$\kappa$, {the higher the value of the robo-advisor relative the investor.}
		}\label{fig:value_added_color_pattern2}
	\end{center}
\end{figure}

We observe from Figure~\ref{fig:value_added_color_pattern2}
that the robo-advisor model yields a higher average cumulative reward over the investment horizon, when compared to the investor-only model. This difference is partly explained by the aforementioned delegation process, which reduces the investor's costs due to the automation of investment decisions. A second contributing factor is that the investor makes mistakes when performing investment decisions. The robo-advisor, instead,
makes portfolio choices that are optimal given her current view of the investor, even if it is limited by the amount of available information on the investor. 
{The left-most bar charts of the figure highlight precisely this effect: the robo-advisor with imperfect knowledge of the investor's risk parameter still outperforms a stand-alone investor who incurs zero costs for making investment decisions herself but makes mistakes and is only correct on average.}

\section{Concluding Remarks and Future Developments} \label{sec:conclusion}

{Robo-advising can substantially enhance human efficiency in investment decisions by handling time-intensive operations. It is crucial, however, that the investor is able to efficiently communicate her preferences to the machine to optimize her objective function. The robo-advisor can only provide a useful service if its valuation of the benefits and risks associated with each action are aligned with the preferences of the investor that it serves.}

We have provided a framework that allows the investor to directly make portfolio choices, and proposed an algorithm through which a robo-advisor can quantify and assess the investor's subjective risk tolerance from her portfolio choices. The provided algorithm allows the robo-advisor to be~$\epsilon$-close to the (intractable) optimal policy with high probability after a polynomial number of steps. Our analysis leverages upon existing results from the PAC-MDP literature, and exploits the structure of the robo-advising problem to obtain sharper bounds than those attained in standard PAC-MDP problems. We have shown that the speed of learning depends on two key properties: consistency of investor's decisions and the required {precision of} the risk-aversion estimate. {We have analyzed the efficiency of the algorithm and the value added by the robo-advisor over the stand-alone investor.} 

In this paper, we have considered a mean-risk approach for portfolio selection where the investor chooses from a pre-specified set of portfolios that lie along the efficient frontier. {An extension of the model would allow the investor to personalize these portfolios based on her personal views. For example, one can employ the Black-Litterman model (\cite{black1990asset}) {to embed the investor's subjective views into the mean-variance Markowitz optimization framework.} 

Our analysis provides a first step towards quantifying the efficiency of learning in the context of risk-preferences estimation.
{The objective function maximized by the robo-advisor favors the construction of algorithms that try to achieve optimality as quickly as possible, with little regard to the costs incurred during the learning period.} This is acceptable in a setting with long investment horizon, such as planning for retirement. If the investment horizon is short, it may instead be desirable to design a less aggressive learning algorithm, i.e., that takes longer to achieve optimality but gains higher total rewards during the learning process.
A suitable performance measure would be the expected decrease in rewards when following the algorithm versus behaving optimally since the beginning.\footnote{This measure is known as \emph{regret}, and was first introduced by \cite{berry1985bandit}.}
We leave such an extension for future research.

{We believe that our framework can be specialized to deal with a larger class of autonomous systems, including those for logistic operations, digital assistance, defense, robotics, and self-driving taxi systems. As an illustrative example, consider the following practical application of our framework. Assume that a network of self-driving taxis tracks each time a user hails a ride. As a part of the service, the rider is able to select one of the several routes or destinations that match a search query; for example, to local restaurants or retail stores. 
	Our framework can generate an assessment of each user's preferences towards various destinations and on the rider's sensitivity towards the risks involved in travel, such as the uncertainty of the arrival time to each destination.
	This assessment enables the taxicab to provide better service to the rider, by presenting options that are customized for the user on the next ride.}

\section*{Funding}
Agostino Capponi's work was supported by the 2019 JPMorgan Chase Faculty Research Award.

\bibliographystyle{apalike}
\bibliography{Bibliography_short}

\newpage
\appendix
\section*{Appendix} \label{sec:appendix}

\section{Proofs}\label{sec:proof}

\begin{proof}[Proof of Lemma~\ref{lemma:escape_bound}]
	For {a} fixed partial path~$h_{L}=(s_1, a_1^{\bM}, a_1^{\bH}, \dots, a^{\bH}_{{j}-1}, a^{\bM}_{{j}-1}, s_{j})$ {with $t<j\leq t+\tau-1$}, let $\Prop(h_{t:j})$ be the probability that the state trajectory~$(s_{t+1}, \dots, s_{j})$ is realized {after time $t$}. Let~${\mathcal{H}_j}$ be the set of all partial paths~$h_{j}$ such that all states~$(s_i)_{{t} < i \leq {j}}$ in that path belong to $K_{t}$ (i.e., {only states with an accurately estimated risk preference are reached}).
	Let~$r^{\pi}_{\theta}({j})$ be the (random) reward received by a type~$\theta$ investor at time~${j}$ given that the robo-advisor policy is~$\pi$, and~$r^{\pi}_{\theta}(h_{j},{j})$ the reward at time~${j}$ given that~$h_{j}$ is the realized partial path. We denote by~$\pi^{\ast}$ the optimal policy and by~$\mathcal{P}_{j}$ the policy followed by Algorithm~\ref{alg:} {at time~$j$}. Then, the following holds:
	\begin{align*}
	\mE \{r^{\pi^{\ast}}_{\theta}({j}) \} - \mE \{r^{\mathcal{P}_{j}}_{\theta}({j}) \}  &= \sum_{{h_j} \in {\mathcal{H}_j}} \Prop(h_{t:j})\left(r^{\pi^{\ast}}_{\theta}(h_{j},{j}) - r^{\mathcal{P}_{j}}_{\theta}(h_{j},{j})\right)  + \\
	& \qquad \qquad \sum_{h_{j} \not\in {\mathcal{H}_j}}   \Prop(h_{t:j})\left(r^{\pi^{\ast}}_{\theta}(h_{j},{j}) - r^{\mathcal{P}_{j}}_{\theta}(h_{j},{j}) \right) \\
	& = \sum_{h_{j} \not\in {\mathcal{H}_j}} \Prop(h_{t:j})\left(r^{\pi^{\ast}}_{\theta}(h_{j},{j}) -  r^{\mathcal{P}_{j}}_{\theta}(h_{j},{j})\right)  \\
	&\leq 2r_{max} \sum_{h_{j} \not\in {\mathcal{H}_j}} \Prop(h_{t:j}) \\
	&\leq 2r_{max} \Prop(A(K_{t},s_{t}))
	\end{align*}
	with a probability at least~$1-\delta$.
	The first step in the above derivation separates the possible paths in which only states with an accurately estimated risk preference are reached from those in which the robo-advisor encounters a state with an  investor risk preference {that is not accurately estimated}. {The first sum equals zero} because~$\mathcal{P}_{j}$ acts optimally on states with {an accurately estimated} risk preference with a probability greater than~$1-|K_{t}|\frac{\delta}{|S|}$ by the union bound. Note that this probability is at least~$1-\delta$ since~$|S|\geq |K_{t}|$. The final inequality makes use of the facts that rewards are bounded~${|r_\theta|\leq r_{max}}$ and the sum of probabilities that a partial path includes~$s \not\in K$ is bounded above by~$\Prop(A(K_{t},s_{t}))$. The result then follows since $V^{{\ast}}_{\tau}(s_{t},\theta) - V^{\mathcal{P}_t}_{\tau}({h_t},\theta)= \sum_{{j=t}}^{t+\tau-1}  \mE \{r^{\pi^{\ast}}_{\theta}({j}) \} - \mE \{r^{\mathcal{P}_{j}}_{\theta}({j}) \}$.

\end{proof}

Before proceeding to the proof of Theorem~\ref{thm:complexity}, we {recall} the following well-known result{, also known as} the Chernoff-Hoeffding Bound.
\begin{lem}\label{lemma:Hoeffding_Bound}
	Suppose a weighted coin, when flipped, has a probability~$p>0$ of landing with heads up. Then, for any positive integer~$k$ and real number~$\delta$, there exists~$m=O\left(\frac{k}{p}\ln\frac{1}{\delta}\right)$, such that after~$m$ tosses, with probability at least~$1-\delta$ we will observe~$k$ or more heads.
\end{lem}

We are now ready to prove Theorem~\ref{thm:complexity}.

\begin{proof}[Proof of Theorem~\ref{thm:complexity}]
	{Let $K_t$ be the set of states with accurately estimated risk preferences at time~$t$.}
	First, suppose that~$\Prop(A(K_{t},s_{t}))\leq \frac{\epsilon}{2\tau r_{max}}$. By Lemma~\ref{lemma:escape_bound}, with probability at least~$1-\delta$, we have
	$$V^{\mathcal{P}_t}_{\tau}({h_t},\theta) \geq V^{\ast}_{\tau}(s_t,\theta) - 2\tau r_{max}\Prop(A(K_{t},s_{t}))\geq V^{\ast}_{\tau}(s_t,\theta) - \epsilon.$$
	Now suppose that~$\Prop(A(K_{t},s_{t}))>\frac{\epsilon}{2\tau r_{max}}$, which implies that the robo-advisor following~$\mathcal{P}_t$ will either follow the optimal policy for the remaining~$\tau$ time steps, or encounter~$s \not\in K_t$ with probability at least~$\frac{\epsilon}{2\tau r_{max}}$. We call the event that the robo-advisor encounters ~$s \not\in K_t$ a ``success''. Then, by Lemma~\ref{lemma:Hoeffding_Bound}, after~$O((\tau^2 \zeta r_{max} /\epsilon)\ln(1/\delta))$ time steps, where~$\Prop(A(K_{t},s_{t}))>\frac{\epsilon}{2\tau r_{max}}$, $\zeta$ successes will occur with a probability at least~$1-\delta$. {In the application of} Lemma~\ref{lemma:Hoeffding_Bound}, the event of a coin landing heads corresponds to the success event after following the robo-advisor's policy for~$\tau$ steps. However, the maximum number of successes that will occur throughout the execution of the algorithm is bounded by $\sum_s C(s)$ and hence $\zeta \leq \sum_s C(s) $. Therefore, by the union bound, with a probability at least~$1-2\delta$ the robo-advisor will execute an~$\epsilon$-optimal policy on all but~$O\left(\frac{\tau^2  r_{max} \sum_s C(s)}{\epsilon}\ln \frac{1}{\delta}\right)$ time steps.
	
	
\end{proof}

\begin{proof}[Proof of Lemma~\ref{lem:complexity_func}]
	Fix a state~$s \in S$. Denote by $\hat{\theta}^{(n)}_s$ the risk-aversion estimator after observing the investor's action~$n$ times in state~$s$. Let~$N=1+\frac{4\sigma_s^2}{\delta \xi^2}$. We have
	\begin{align*}
	\Prop(|\theta_s - \hat{\theta}^{(N)}_s |<\xi/2) &\geq 1-\frac{4\sigma_s^2}{N\xi^2} \\
	&\geq 1-\delta.
	\end{align*}
	where the first line follows by Chebyshev's inequality and the second by the fact $N=1+\frac{4\sigma_s^2}{\delta \xi^2}>\frac{4\sigma_s^2}{\delta \xi^2}$.
	In other words, with a probability of at least~$1-\delta$, we have $\theta_s = \argmin_{x \in \Theta} |x - \hat{\theta}^{(N)}_s  | $.

\end{proof}

\begin{proof}[Proof of Lemma~\ref{lem:complexity_func_R}]
	Fix a state~$s \in S$. Denote by $\hat{\theta}^{(n)}_s$ the risk-aversion estimator after observing the investor's action~$n$ times in state~$s$. Let~$N=\frac{2R_s}{\xi^2}\ln \frac{2}{\delta}$. We have
	\begin{align*}
	\Prop(|\theta_s - \hat{\theta}^{(N)}_s |<\xi/2) &> 1-2\exp(-\frac{N\xi^2}{2R_s^2}) \\
	&= 1-\delta.
	\end{align*}
	where the first line follows from Hoeffding's inequality and the second by choosing~$N=\frac{2R_s}{\xi^2}\ln \frac{2}{\delta}$.
	In other words, with a probability of at least~$1-\delta$, we have $\theta_s = \argmin_{x \in \Theta} |x - \hat{\theta}^{(N)}_s  |$.
\end{proof}

\section{Incremental formula}\label{app:incremental}
In this section, we present the detailed derivation of the incremental formula given in Section~\ref{sec:alg}.
\begin{align*}
\hat{\theta}^{(n)}_s	&=\frac{1}{n}\sum_{i=1}^{n} g^{-1}(a_{i,s}^{\bH};s) \\
&=\frac{1}{n}\left(g^{-1}(a_{n,s}^{\bH};s)+\sum_{i=1}^{n-1}g^{-1}(a_{i,s}^{\mathbf{H}};s) \right) \\
&=\frac{1}{n}\left(g^{-1}(a_{n,s}^{\bH};s) + (n-1)  \frac{1}{n-1}\sum_{i=1}^{n-1}g^{-1}(a_{i,s}^{\bH};s)\right) \\
&=\frac{1}{n}\left(g^{-1}(a_{n,s}^{\bH};s) + (n-1) \hat{\theta}_s^{(n-1)}\right)	\\
&=\frac{1}{n}\left(g^{-1}(a_{n,s}^{\bH};s) + n\hat{\theta}_s^{(n-1)}-\hat{\theta}_s^{(n-1)}\right) \\
&=\hat{\theta}_s^{(n-1)}+\frac{1}{n}\left(g^{-1}(a_{n,s}^{\bH};s) -\hat{\theta}_s^{(n-1)}\right)
\end{align*}

\section{The Infinite-Horizon Case}\label{app:infinite_horizon}
In this section, we extend our results to the case where $T=\infty$ and introduce a discount factor~$\gamma<1$. Let~$r^{\pi}_{\theta}(t)$ be the (random) reward received by a type~$\theta$ investor at time~$t$ given that the robo-advisor policy is~$\pi$. Let $V^{\pi}(s_{t},\theta):=\Ex\left(\sum_{{j=1}}^\infty \gamma^{{j}-1} r^{\pi}_{\theta}({j}) \right)$ denote the discounted, infinite-horizon value function corresponding to a policy~$\pi$, starting from state~$s_{t}$ and given that the risk {aversion} parameter of the investor is~$\theta$. If~$T$ is a positive integer, let $V^{\pi}_T(s_{t},\theta):=\Ex\left(\sum_{{j}=1}^T \gamma^{{j}-1} r^{\pi}_{\theta}({j}) \right)$ denote the $T$-step value corresponding to policy~$\pi$, when the market starts from state~$s_{t}$ and the risk aversion parameter of the investor is~$\theta$. {If $\pi$ is non-stationary, then $s_{t}$ is replaced by the partial path~$h_t$ in the previous definitions.} We then have the following extension of Lemma~\ref{lemma:escape_bound}.

\begin{lem}\label{lemma:escape_bound2}
	{Let~$\mathcal{P}_t$	denote the policy followed by Algorithm~\ref{alg:} at time~$t$, $s_t$ be the economic state at $t$,} and set $C(s)=f\left(s,\frac{\delta}{|S|}\right)$.
	Let $\theta$ be the true risk aversion parameter of the investor, $K_{t}$ the set of states with {accurately estimated} risk preferences, and~$A(K_{t},s_{t})$ be the event that a state $s \notin K_{t}$ is {visited} in the following~$T$ steps if the market starts in state~$s_{t} \in K$. Then, with probability at least~$1-\delta$,
	$$V^{\mathcal{P}_t}_{T}({h_t},\theta) \geq V^{\ast}_{T}(s_{t},\theta) - \frac{2r_{max}\Prop(A(K_{t},s_{t}))}{1-\gamma} $$
	where~$|r(\cdot)| \leq r_{max}$ and~$\mathcal{P}_t$ is the policy followed by Algorithm~\ref{alg:} at time~$t$.
\end{lem}
\begin{proof}
	Replacing~$\tau$ with~$T$, and proceeding as in the proof of Lemma~\ref{lemma:escape_bound}, we get
	\begin{align*}
	V^{\mathcal{P}_t}_{T}({h_t},\theta) &\geq V^{\ast}_{T}(s_{t},\theta) - 2r_{max}\Prop(A(K_{t},s_{t}))\sum_{{j=1}}^T \gamma^{t-1} \\
	&\geq V^{\ast}_{T}(s_{t},\theta) - \frac{2 r_{max}\Prop(A(K_{t},s_{t}))}{1-\gamma}
	\end{align*}
	where the last inequality follows from the fact that~$\sum_{{j}=1}^T \gamma^{{j}-1}=\frac{1-\gamma^T}{1-\gamma}<\frac{1}{1-\gamma}$ for $\gamma<1$.
\end{proof}

Next, we prove the extension of Theorem~\ref{thm:complexity} to the infinite horizon discounted case.

\begin{thm}\label{thm:complexity2}
	Let~$\mathcal{P}_t$	denote the policy followed by Algorithm~\ref{alg:} at time~$t$, $s_t$ be the {economic} state at $t$, and {set} ~$C(s)=f\left(s,\frac{\delta}{|S|}\right)$.
	Let~$\theta$ be the true investor type. Then with probability at least~$1-2\delta$,
	$$V^{\mathcal{P}_t}({h_t},\theta) \geq V^{\ast}(s_t,\theta) - 3\epsilon $$
	--- i.e., the algorithm is $3\epsilon$-close to the true optimal policy --- for all but
	$$ m= O\left(\frac{ r_{max}\sum_s C(s)}{\epsilon(1-\gamma)^2}\ln\frac{1}{\delta}\ln \frac{r_{max}}{\epsilon(1-\gamma)}\right) $$
	time steps.

\end{thm}
\begin{proof}
	Let $T=\frac{1}{1-\gamma}\ln \frac{r_{max}}{\epsilon(1-\gamma)}$. It follows from Lemma~2 of \cite{kearns2002near} that $|V^{\pi}_T(s_{t}, \theta)-V^{\pi}(s_{t}, \theta)|\leq \epsilon$ for any state $s_{t}$ and policy~$\pi$. {The previous inequality also holds for non-stationary policies.} {Let $K_t$ be the set of states with accurately estimated risk preferences at time~$t$.}
	First, suppose that~${\Prop(A(K_{t},s_{t}))\leq \frac{\epsilon(1-\gamma)}{2 r_{max}}}$. Then, with probability at least~$1-\delta$,
	\begin{align*}
	V^{\mathcal{P}_t}({h_t},\theta) &\geq V^{\mathcal{P}_t}_{T}({h_t},\theta) -\epsilon \\
	&\geq V^{\ast}_{T}(s_{t},\theta) - \frac{2r_{max}\Prop(A(K_{t},s_{t}))}{1-\gamma} -\epsilon \\
	& \geq V^{\ast}_{T}(s_{t},\theta) - 2\epsilon \\
	& \geq  V^{\ast}(s_{t},\theta) - 3\epsilon,
	\end{align*}
	where the first and last inequalities follow from the above definition of $T$, the second inequality by Lemma~\ref{lemma:escape_bound2}, and the third inequality by the assumption ${\Prop(A(K_{t},s_{t}))\leq \frac{\epsilon(1-\gamma)}{2 r_{max}}}$. 
	Next, suppose that~$\Prop(A(K_{t},s_{t}))>\frac{\epsilon(1-\gamma)}{2 r_{max}}$, which implies that the robo-advisor implementing~$\mathcal{P}_t$ will either follow the optimal policy for the next~$T$ time steps, or encounter a state~$s \not\in K_t$ with probability at least~$\frac{\epsilon(1-\gamma)}{2 r_{max}}$. We call the event that the robo-advisor encounters~$s \not\in K_t$ a ``success''. Then, by Lemma~\ref{lemma:Hoeffding_Bound}, after~$O(\frac{T \zeta r_{max}}{\epsilon(1-\gamma)} \ln(1/\delta))$ time steps, where~$\Prop(A(K_{t},s_{t}))>\frac{\epsilon(1-\gamma)}{2 r_{max}}$, $\zeta$ successes will occur with a probability at least~$1-\delta$. {In the application of} Lemma~\ref{lemma:Hoeffding_Bound}, the event of a coin landing heads corresponds to the success event after following the robo-advisor's policy for~$T$ steps. However, the maximum number of successes that will occur throughout the execution of the algorithm is bounded by $\sum_s C(s)$ and hence $\zeta \leq \sum_s C(s) $. Therefore, by the union bound, with a probability at least~$1-2\delta$ the robo-advisor will execute an~$3\epsilon$-optimal policy on all but~$O\left(\frac{Tr_{max} \sum_s C(s) }{\epsilon(1-\gamma)}\ln\frac{1}{\delta}\right)=O\left(\frac{r_{max}\sum_s C(s) }{\epsilon(1-\gamma)^2}\ln\frac{1}{\delta}\ln \frac{r_{max}}{\epsilon(1-\gamma)}\right)$ time steps.
\end{proof}

\section{Resolution Parameter}\label{app:granularity}
In this section, we study how different values of the grid resolution parameter~$\xi$ affect the rewards. We assume that the true investor's risk-aversion parameter is uniformly sampled from the interval~$[2.2,8.3]$, as {discussed in the model calibration} Section~\ref{sec:empirical_calibration}. 
Figure~\ref{fig:xi} shows, {consistently with intuition}, that as the resolution parameter decreases, the rewards improve. When the resolution parameter is small, the robo-advisor is able to estimate the true risk aversion parameter more accurately, and hence can make more customized portfolio choices. The figure shows that the benefits obtained from reducing the resolution parameter below~$\xi=0.1$ are insignificant.

\begin{figure}[H]
	\begin{center}
		\includegraphics[scale=1]{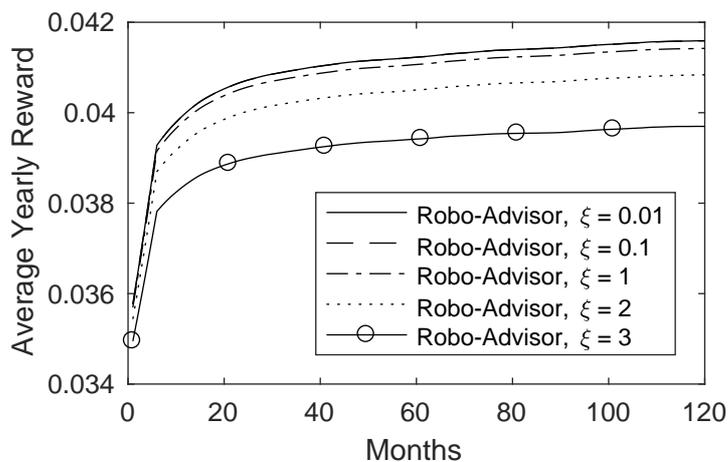}
		\caption{Yearly reward versus time, averaged over $10^4$ trials. The range of mistakes~${r=3}$, the sample complexity~$C=5$, the intervention cost~$\kappa=0.08\%$, and the total learning period is $10$ years (i.e., 120 months). Reducing the resolution parameter~$\xi$ improves the estimate of the risk aversion parameter~$\theta$ and hence enables the robo-advisor to make more tailored investment decisions, improving investor rewards. The value of reducing the resolution parameter below~$\xi=0.1$ is minuscule. }  \label{fig:xi}
	\end{center}
\end{figure}

\end{document}